\documentclass[12pt,a4paper,normalheadings,headsepline,headinclude,bibtotoc,fleqn,review]{article}
\usepackage[latin1]{inputenc}
\usepackage{amsmath,array,amsthm,graphicx,amssymb,natbib,color,mathtools,xspace,enumitem}
\usepackage[lines=10]{geometry}

\newtheorem{prop}{Proposition}

\usepackage{fancyhdr} 
\begin{document}

\begin{center}\huge{Observing Each Other's Observations in the Electronic Mail Game}\footnote{We thank Martin Hellwig
for comments and for spotting an error. We also received helpful
comments and questions from Brian Cooper and seminar participants
in Bonn.}
\end{center}
\begin{center}\emph{Dominik Grafenhofer}\\\emph{Max Planck Institute for Research on
Collective Goods, Kurt-Schumacher-Str. 10, 53113 Bonn, Germany.
Email: grafenhofer@coll.mpg.de.}
\end{center}
\begin{center}\emph{Wolfgang Kuhle}\\\emph{Max Planck Institute for Research on
Collective Goods, Kurt-Schumacher-Str. 10, 53113 Bonn, Germany.
Email: kuhle@coll.mpg.de.}
\end{center}

\noindent\emph{\textbf{Abstract:} We study a Bayesian coordination
game where agents receive private information on the game's payoff
structure. In addition, agents receive private signals on each
other's private information. We show that once agents possess
these different types of information, there exists a coordination
game in the evaluation of this information. And even though the
precisions of both signal types is exogenous, the precision with
which agents predict each other's actions at equilibrium turns out
to be endogenous. As a consequence, we find that there exist
multiple equilibria if the private signals' precision is high.
These equilibria differ with regard to the way that agents
weight their private information to reason about each other's actions.}\\
\textbf{Keywords: Coordination Games, Equilibrium Selection, Primary Signals, Secondary Signals}\\
\textbf{This version: \today} \vspace{.5cm}


\section{Introduction}\label{intro}

It is the nature of coordination games that one player's payoff
depends not only on his own action, but also on the action of
others. And if players can observe each other's actions, this
interdependence induces multiple equilibria. The same is not
necessarily true in environments where agents have only
probabilistic knowledge over each other's actions. \citet{Har73}
and \citet{Sel75} develop frameworks where one player cannot
observe the other player's payoffs, respectively, where players
may ``tremble", and show that the number of equilibria can be
drastically reduced. Such imperfect knowledge regarding the other
player's actions arises naturally in the games of \citet{Rub89}
and \citet{Car93}, where both players receive noisy private
information over an unknown coefficient that characterizes the
game's payoffs. In such models, players face payoff uncertainty as
in \citet{Har73}. At the same time, due to their imperfect private
knowledge of the payoffs, they tremble as in \citet{Sel75}.
Moreover, the models of \citet{Rub89} and \citet{Car93} predict
that players will play a unique risk-dominant equilibrium rather
than the multiple equilibria that they would play if actions and
payoffs were common knowledge. The present paper studies this
selection mechanism in the context of a more general information
structure. We show that once agents possess different types of
private information, there exists a coordination game in the
evaluation of this information. And even though the precisions of
both signal types is exogenous, the precision with which agents
predict each other's actions at equilibrium turns out to be
endogenous. As a consequence, there exist multiple equilibria if
the private signals' precision is high.

More precisely, the players of \citet{Rub89} and \citet{Car93}
reason only indirectly about each other. They receive a signal
regarding the game's payoffs, which they use to update their
beliefs regarding the game's coefficients. Second, knowing that
the other agent's signal is \emph{correlated} with their own, they
use their signal to infer the other player's posterior beliefs. In
the present paper, we argue that players usually know more than
that. Often players can observe \emph{directly} parts of the other
player's observation. That is, in the context of the coordinated
attack interpretation of the \citet{Rub89} ``electronic mail
game'', the general may observe that his messenger is off to a
``good start" since he already managed to cross the first valley.
Thus, the chances that he eventually arrives at the other camp are
better than usual. Put differently, the sender of the ``primary
message" knows that the other player most likely received a
message indicating that a particular game was chosen. Or,
alternatively, one general can observe from a distance that
someone entered the other general's camp. While he cannot be sure
whether it was his messenger or someone else, this observation
induces him to revise upward the probability with which the other
general received the news. On the other hand, if the sender sees
that the messenger is off to a bad start, then he knows that it is
less likely that the message will reach its receiver. The main
result of the paper shows that such a generalized information
structure induces multiple equilibria. If one player conditions
his actions on the \emph{``primary signal"}, which informs him of
the particular game chosen by nature, the other player will have
an incentive to condition on his \emph{``secondary signal"}, which
informs him of the \emph{``primary signal"} that the other player
received, and vice versa. Such asymmetric evaluation of signals
will maximize the precision with which players can anticipate each
other's actions. That is, even though the precision of both signal
types is exogenous, the precision with which agents anticipate
each other's actions at equilibrium turns out to be endogenous.
And there will exist multiple equilibria that differ regarding the
way that agents weight the different types of private information
that they receive, to reason about each other's actions.

In the present note, we formalize this intuition in the context of
the \citet{Rub89} electronic mail game. First, we introduce the
basic information structure, where agents purely rely on the
correlation of their private observations to infer the other
agent's beliefs. That is, in equilibrium, agents rely on the fact
that they did not receive a confirmation of their last message,
which may mean either that their last message did not reach the
receiver, or that the receiver's reply was lost. Within this
setting, we recall the main insight, namely that agents play a
unique risk-dominant equilibrium. In the following, we refer to
the basic signals from the electronic mail game as \emph{primary
signals}. In the main part of this note, we introduce a
\emph{secondary signal} which allows agents to make additional
inference on the other agent's observations. That is, we introduce
a noisy signal that allows players to infer \emph{directly} the
probability with which their signal reached the receiver. For the
modified setting, we find that if players observe each other's
signal with great precision, then they can coordinate on multiple
equilibria.

To interpret our findings, we compare two classes of
equilibria: (i) symmetric equilibria and (ii) asymmetric
equilibria. The distinguishing feature of a symmetric equilibrium
will be that agents weight their two signal types equally. In an
asymmetric equilibrium one agent leans heavily on the
\emph{secondary signal}, while the other agent has an incentive to
lean heavily on his \emph{primary signal} and vice versa. The
key feature is therefore that the two signals are
``cross complements''. That is, if player one relies heavily on
his \emph{secondary signal}, then the other player has an
incentive to rely on his \emph{primary signal}. Such an asymmetric
weighting of private signals enables agents to maximize the
precision with which they can forecast each other's actions.
Finally, to emphasize the importance of the class of asymmetric
equilibria, we show that asymmetric equilibria dominate symmetric
equilibria on efficiency grounds.

The main contribution of the present paper is the introduction of
a new class of private signals. Namely, signals about the other
player's signals. Moreover, we show that such information can
ensure multiple equilibria once private signals are sufficiently
precise. Compared to the literature, we note that \citet{Rub89},
\citet{Car93}, and \citet{Fra04} have studied two-action
coordination games, where agents receive what we call
\emph{primary signals} that allow them to make inference on the
game's unknown coefficients.\footnote{See \citet{Fra04} for a
broad literature overview on equilibrium selection through what we
call \emph{primary signals}. \citet{Car93}, pp. 1008-1010, for a
detailed comparison of their ``global games", which rely on
continuous distributions, with the ``electronic mail game" and its
discrete information structure.} Moreover, through the correlation
of private information, agents can reason about each other's
posteriors and actions. Regarding equilibrium selection, these
studies predict that unique equilibria are ensured once private
signals are sufficiently precise. The present example shows that
the existence of secondary \emph{private} signals can invert this
finding: multiple equilibria are ensured once the private signals
are sufficiently precise.

Regarding different types of signals, \citet{Mor04},
\citet{Hel02}, \citet{Met02}, and \citet{Ang06} emphasize the role
of public signals/common priors in the global games framework,
showing that such signals restore equilibrium multiplicity if
public signals are sufficiently precise compared to the private
signal; we give an example where multiplicity arises in pure
private signal environments. A further class of signals was
introduced by \citet{Min03}, \citet{Ang06}, and \citet{Das07}, who
study environments where agents observe each other's actions. Such
signals tend to induce unique equilibria in the two-player games
of \citet{Min03}, where signals over each other's actions are
perfectly revealing. \citet{Ang06}, and \citet{Das07} study public
signals that partially reveal the other player's actions. They
show that multiplicity may emerge if the public signal is of high
quality. \citet{Kuh13} gives an example where the public signal's
quality reduces the number of equilibria. Finally, \citet{Rub89}
points out that equilibrium multiplicity may reemerge once there
is a technical upper bound for the number of exchanged messages.
Similarly, multiplicity also obtains in the model of
\citet{Bin01}, where agents can \emph{decide} whether or not to
send electronic messages which are costly.

The rest of the paper is organized as follows. Section \ref{emg}
outlines our electronic mail game. In Section \ref{no_z_info}, we recall the uniqueness result for our modified game without secondary signals. Section \ref{with_z_info} contains the main result. Sections \ref{d} concludes.


\section{A symmetric electronic mail game}\label{emg}


There are two players $1$ and $2$. Each has two actions $A$ and
$B$ to choose from. And there is uncertainty about which game
$G_a$ or $G_b$ the two players are going to play. Games $a$ and
$b$ differ regarding their payoffs. Nature selects game $a$ with
probability $1-p$ and game $b$ with probability $p<\frac12$. The
game's payoff structures are:
$$\stackrel{\text{\normalsize Game $G_a$}}{\begin{array}{r|cc}
 & A & B \\ \hline
A & M,M & 0,-L\\
B & -L,0 & 0,0\\
\end{array}} \qquad\qquad
\stackrel{\text{\normalsize Game $G_b$}}{\begin{array}{r|cc}
 & A & B \\ \hline
A & 0,0 & 0,-L\\
B & -L,0 & M,M\\
\end{array}}$$
Moreover, we assume $L>M>0$. Hence, players face a coordination
problem in both states of the world: if players coordinate on
actions $A$ ($B$) in state $a$ ($b$), they receive $M$ each, while
coordination on $B$ ($A$) yields $0$ to both players. However, if
players fail to coordinate, i.e., choose different actions, then
the player who plays $B$ receives $-L$, and the payoff for playing
$A$ is $0$. Players receive private information on the game's
fundamental before they choose an action. The probability $p$, the
payoff structure, and the forthcoming \emph{communication
protocol} are common knowledge among players.

In state $a$, both players get information $T_1=T_2=0$. In state
$b$, \emph{one player is randomly selected} with probability
$\frac12$, and informed of the true state $b$. 
The selected player $i$ then sends a message to player $j$. The
message, however, is lost with probability $\varepsilon$. Upon
receiving a message, player $j$ sends a confirmation back to
player $i$ which is also lost with probability $\varepsilon$.
These messages are exchanged until finally one message is lost and
communication ends. Players 1 and 2 now choose their actions based
on the number of messages $T_1$ and $T_2$ that they received.

The present game therefore differs from the \citet{Rub89} game in
that it is random which player spots the actual game selected, and
starts to inform the other player. Moreover, we assume that both
players do not know who was selected to send the first
message.\footnote{In Appendix \ref{ref}, we derive our main
result for the original \citet{Rub89} game, where player 1 is
informed of the state of nature with probability 1.}
This symmetric structure accommodates a more natural
interpretation of the \emph{asymmetric} equilibria that players
play once we introduce \emph{secondary signals} that inform
players about each other's \emph{primary signals} $T_1$ and $T_2$.
That is, to bring out the source of multiple asymmetric
equilibria, we endow players with signals of identical precisions
and, unlike \citet{Rub89}, assume that the probability with which
they observe the true state of nature, is the same $(1/2)$ across
players. The propositions in this paper regarding the existence of
multiple asymmetric equilibria continue to hold once we set the
probability with which player 1 observes the true state of nature
to one as in \citet{Rub89}. Finally, in the context of the
coordinated attack interpretation of the electronic mail game, it
seems natural that players do not necessarily know that ``player
1" always observes the true state of nature first. As such, the
present specification may be seen as less restrictive.

\subsection{No direct information on the other player's
information} \label{no_z_info}

Before turning to our main result, we establish the uniqueness result
of \cite{Rub89} for our symmetric mail game.

\begin{prop}\label{sym_game}
There exists only one equilibrium in which player 1 plays $A$ in
the state of nature $a$. In this equilibrium, both players play
$A$, irrespective of the number of received messages $T_1$ and
$T_2$. \label{rub}\end{prop}
\begin{proof}See Appendix \ref{A0} for the proof by forward
induction.\end{proof}

Proposition \ref{rub} recalls the inductive equilibrium selection
mechanism that operates through higher-order beliefs: If player 1
plays always $A$ for $T_1=0$, then player 2 also plays $A$, and
this equilibrium procreates to games, where $T_i>0, i=1,2$. That
is, even though both players $T_i>0, i=1,2$ know that game $b$ was
selected, players still play $(A,A)$, despite the fact that $(B,B)$
would be payoff-dominant. However, as in \citet{Rub89}, there
exists a second equilibrium, where both players play $(B,B)$,
receiving a zero payoff both all the time. This equilibrium can be removed once
the $(B,B)$ payoff in game $G_a$ is negative, rather than 0, for
both players. In this case, there exists only one unique
equilibrium, where both players play $A$. Such a modification of
payoffs, which may be introduced throughout the paper, would bring
us closer to the formulation of \citet{Car93}, where there exist
unique strict equilibria for certain signal values.\footnote{See
\citet{Car93}, pp. 1008-1010, for a detailed discussion of the
relation between global games and mail games.}



\section{Information on the other player's information} \label{with_z_info}

Let us now add a secondary signal $Z_1$ and $Z_2$ as another
source of private information: player $i$ not only gets information $T_i$ but also observes
$$Z_i:=\left\{\begin{array}{ll}
T_j & \qquad\text{with probability } 1 - \psi\\
T_j + 1 & \qquad\text{with probability } \psi \, .
\end{array}\right.$$
The secondary signal $Z_1$ informs player 1 of the primary signal
$T_2$ that player 2 received. As such, the secondary signal
carries two types of information. First, it allows player 1 to
reason about the true fundamental of the game. That is, through
its dependence on $T_2$, $Z_1$ is correlated with nature's choice
of a fundamental. Second, and more importantly, $Z_1$ allows
player 1 to look \emph{directly} at $T_2$. In turn, this direct
look at $T_2$ informs him about the probability with which player
2 plays $A$ or $B$. In the following main propositions
\ref{gen_asym_equ}, \ref{p5} and \ref{p6}, we show that this
``direct look" at the other player's signal will induce asymmetric
equilibria, in which players weight their signals $Z$ and $T$
differently. That is, if player 1 conditions his actions mainly on
his primary signal $T_1$, then player 2 will have an incentive to
weight signal $Z_2$ heavily and vice versa. Put differently, the
signals $T_i,Z_j$ are complements, while the signals
$T_i,T_j$ are substitutes.

To underscore the significance of these asymmetric equilibria, we
proceed in three steps. First, we show that they exist. Second, we
describe the symmetric equilibria, where agents weight their
signals symmetrically. Third, we show that the asymmetric
equilibria welfare-dominate symmetric equilibria. Before we study
the asymmetric equilibria, we note that the \citet{Rub89}
equilibrium carries over to the environment where agents receive
primary and secondary signals.
\begin{prop} \label{gen_orig_equ}
When information $Z_1$ and $Z_2$ are available to players, there
exists an equilibrium in which both players play $A$ irrespective
of the information received.\end{prop}
\begin{proof}
Suppose player $1$ thinks that player $2$ plays $A$ for sure.
Irrespective of $(T_1,Z_1)$, the following holds: Choosing $B$
will yield a payoff $-L$, while taking action $A$ will secure him
a payoff of $M$. The same argument can be made for player $2$, and
thus we have established that the strategy profile $(A,A)$ is an
equilibrium.
\end{proof}
In this equilibrium, both players receive a zero payoff, even in those
situations where they know that playing $(B,B)$ would yield a higher
payoff. However, players can use their private signals to
coordinate on an alternative class of equilibria:

\begin{prop}\label{gen_asym_equ} If the secondary signals'
precision is sufficiently high ($\psi$ sufficiently small), there exist two asymmetric threshold equilibria for every $n\in\{1,2,3,\dots\}$. In one
equilibrium, player 1 plays $B$ if and only if $T_1\geq n+1$ and
$Z_1\geq n$; player 2 plays $B$ if and only if $Z_2\geq n+1$ and
$T_2 \geq n$. In the second equilibrium player $1$ plays $B$ if
and only if $Z_1\geq n+1$ and $T_2\geq n$ and player 2 plays $B$
if and only if $T_2\geq n+1$ and $Z_2\geq n$.
\end{prop}
\begin{proof}
Let us consider the first equilibrium with cutoff $n$.
\begin{enumerate}
\item Take the behavior of player 2 as given. There are three cases to consider:
\begin{enumerate}
\item $T_1<n$: Player 1 is sure that $Z_2\leq n$ and hence plays
$A$. \item $T_1=n$: With probability $1-\psi$ ($\psi$) player 2's
information is $Z_2=n$ ($Z_2=n+1$). Playing $A$ secures a payoff
of zero for sure; playing $B$ yields an expected payoff larger
than $(1-\psi) (-L)+\psi M$, which is the first player's payoff
from $B$, when player 2 always plays $B$ given $Z_2>n$. Thus, for
$\psi\leq \frac{L}{L+M}=: \psi_1$ playing $A$ is optimal. \item
$T_1\geq n+1$:  Player 1 is sure that $Z_2 \geq n+1$ and $T_2\geq n$,
hence finds it optimal to play $B$.
\end{enumerate}
\item Equivalently, now take the behavior of player 1 as given.
\begin{enumerate}
\item $Z_2\leq n$: Player 2 knows that $T_1\leq n$, and thus plays $A$.
\item $Z_2 > n+1$: Player 2 knows that $T_1\geq n+1$, and thus plays $B$.
\item $Z_2 = n+1$: Here we have to take care of four subcases:
\begin{enumerate}
\item $T_2 = n-1$: Hence $T_1=n$ for sure and player 2 thus
chooses $A$.

\item $T_2 = n$: Defining $\lambda_\psi:=P(T_1\leq n | T_2 = n
\wedge Z_2=n+1)=\frac{\psi}{\psi+\frac{1-\epsilon}2 (1-\psi)}$,
the payoff for playing $B$ can be written as $\lambda_\psi (-L)+
(1-\lambda_\psi) M$. From this we obtain a boundary
$\psi_2:=\frac{(1-\epsilon)M}{2L-(1-\epsilon)M}>0$, which ensures
that for all $\psi\leq\psi_2$ playing $B$ is optimal for player 2.
That is, for $\psi\leq\psi_2$ the expected payoff of playing $B$
is non-negative.

\item $T_2 = n+1$: We can repeat the same argument using
$\mu_\psi:=P(T_1\leq n | T_2 = n+1 \wedge
Z_2=n+1)=\frac{\psi}{\psi+(1-\epsilon)(1-\psi)}$. It holds that
$\mu_\psi<\lambda_\psi$, such that for all $\psi\leq\psi_2$
playing $B$ is optimal for player 2. \item $T_2 = n+2$: Hence
$T_1=n+1$ for sure, and player 2 chooses $B$.
\end{enumerate}
\end{enumerate}
\end{enumerate}
Again, we can choose $\psi$ sufficiently small, i.e. $\psi \leq \min
\{\psi_1,\psi_2\}$, such that the strategy profile from the
proposition is indeed an equilibrium.

The second part of the proposition follows immediately from
changing the roles of player $1$ and $2$.
\end{proof}

To interpret the equilibria in Proposition \ref{gen_asym_equ} we
note that players weight primary and secondary signals
asymmetrically. That is, if player $1$ switches from playing $A$
to playing $B$ for signal pairs $T_1\geq n+1, Z_1\geq n$, then
player $2$ switches from $A$ to $B$ for signal values $T_2\geq n,
Z_2\geq n+1$. And, as the proof shows, signals where the trigger
strategy requires values greater or equal $n+1$ carry the main
information regarding the other player's signals and actions. On
the contrary, signals where the trigger strategy requires values
greater or equal $n$ carry little information on other player's
signals. More precisely, player $1$ relies in his inference about
the other player's action on the fact that $T_1\geq n+1$ informs
him of the fact that $T_2\geq n, Z_2\geq n+1$. Hence, player $1$
relies on his primary signal to infer the action of player $2$.
The main reason for player 1's reliance on his primary signal
$T_1$, is that player $2$ conditions his actions on $T_2\geq n,
Z_2\geq n+1$. That is, as the steps 2.(c)$i-iv$ in the proof
show, player 2 relies on his secondary signal to infer the action
of player 1. In turn, player 1's reliance on the secondary signal
1 justifies player 2's reliance on the primary signal... This
complementarity between player 1's primary and player 2's
secondary signal ensures that asymmetric weighting of signals is
an equilibrium. Put differently, players face a coordination game
in the weighting of their private signals; players can
\emph{choose their cutoff values $T_i,Z_i$ in a way that makes it
easy for their counterpart to assess whether their requirement for
playing $B$ is met or not}. In the present case, this means
leaning on the primary signal once the opponent leans on the
secondary signal and vice versa.

The main purpose of the following propositions
\ref{gen_sym_equ}-\ref{p6} is to emphasize this point further.
First, we show that there also exist symmetric equilibria, where
agents weight their signals equally. Moreover, we show that not
every configuration of cutoffs is an equilibrium. Second, we show
in proposition \ref{p5} that the coordination game in the
evaluation of private signals ensures multiple equilibria once
private signals are of high quality. Finally, to bring-out the
fact that the precision with which agents can anticipate each
others actions is endogenous, we show in proposition \ref{p6} that
the asymmetric equilibria, where agents exploit the
complementarity between primary and secondary signals, welfare
dominate the symmetric equilibria of proposition
\ref{gen_sym_equ}.

\begin{prop}\label{gen_sym_equ} If the secondary signals'
precision is sufficiently high, there exist symmetric monotone equilibria for every $n\in\{1,2,3,\dots\}$, where players weight their signals equally
such that both players play $B$ if and only if $T_i \geq n+1 $ and $Z_i
\geq n+1$. There exist no symmetric monotone equilibria, where both
players play $B$ if and only if $T_i \geq n+1 $ and $Z_i \geq n+2$ (or
$T_i \geq n+2 $ and $Z_i \geq n+1$).
\end{prop}
\begin{proof}See Appendix \ref{A1}.\end{proof}

Up to now we have shown that our main results hold for a given
probability $\epsilon$ that a message between the players is lost.
One might suspect\footnote{In the public and private information
frameworks of \citet{Hel02}, \citet{Mor04}, and \citet{Ang06},
multiple equilibria emerge once public signals or priors are
sufficiently precise \emph{relative} to private signals.} that
equilibrium multiplicity depends on the relative precision of
primary and secondary signals, i.e., a high $\epsilon/\psi$ ratio
may be required. The following proposition shows that this is not
the case:

\begin{prop}
There exist upper bounds $\bar{\epsilon}>0$ and $\bar{\psi}>0$, such
that the equilibria described in propositions \ref{gen_asym_equ}
and \ref{gen_sym_equ} exist for all combinations of  $\epsilon\leq\bar{\epsilon}$ and
$\psi\leq\bar{\psi}$. \label{p5}\end{prop}
\begin{proof} For propositions \ref{gen_asym_equ}
and \ref{gen_sym_equ} to hold, we need a sufficiently small error
probability for the secondary signal, i.e.,
$\psi\leq\min[\psi_1,\psi_2,\psi_3]$, where
$\psi_1=\frac{L}{L+M}$,
$\psi_2=\frac{(1-\epsilon)M}{2L-(1-\epsilon)M}$, and $\psi_3=
\frac{(1-\epsilon)M}{L+(1-\epsilon)M}$. It therefore suffices for
show that the limits of $\psi_1,\psi_2,\psi_3$ for $\epsilon
\rightarrow 0$ are positive: First, observe that $\psi_1$ is positive and does not depend on $\epsilon$. Second, $\lim\limits_{\epsilon\rightarrow 0}\psi_2 = \frac{M}{2L-M} > 0$. Finally, $\lim\limits_{\epsilon\rightarrow 0}\psi_3 = \frac{M}{L+M} > 0$.
\end{proof}
This demonstrates, that our results differ fundamentally from
those obtained by \citet{Rub89}, \citet{Car93}, \citet{Fra04}, and
\citet{Mor07}, where equilibrium selection works best once private
information is very precise. The present analysis, in particular
Proposition \ref{gen_asym_equ}, shows that once agents receive
different types of information there emerges a coordination game
in the evaluation of this information. And this incentive to
coordinate is particularly strong once private signals are very
informative. 

In the introduction, we argued that the asymmetric equilibria
derived here deserve special scrutiny. We believe that the main
reason for this lies in the following

\begin{prop} If the secondary private signals are very precise, the asymmetric equilibria described in Proposition
\ref{gen_asym_equ} for $n=1$ welfare-dominate those where $n>1$.
And, more importantly, the asymmetric equilibria of Proposition
\ref{gen_asym_equ} welfare-dominate the symmetric ones of
Proposition \ref{gen_sym_equ} for every given cutoff $n$.
\label{p6}\end{prop}
\begin{proof}See Appendix \ref{A2}.\end{proof}

Proposition \ref{p6} emphasizes the complementarity between
primary and secondary signal that gave rise to the asymmetric
equilibria of proposition \ref{gen_asym_equ}: Private signals of
high quality induce  for players to coordinate on an equilibrium,
where they exploit the complementarity between primary and
secondary signals to reduce losses that occur in cases where
nature selects game $b$, but players play $(A,A)$, or, worse, $(A,B)$.

\section{Discussion}\label{d}


The players of
\citet{Rub89} and \citet{Car93} rely on a very particular type of
information. Players can only reason \emph{indirectly} about each
other's trembling behavior because the private information on the
game's unknown coefficient, is correlated. That is, players purely
rely on the knowledge that the other player is looking at the same
game. They cannot make \emph{direct} inference on what the other
player thinks of the game. In reality, we argue that there are
many environments where players know more than that. In the
coordinated attack interpretation of the mail game, players may
observe directly that their messenger is off to a ``bad start", in
which case it is unlikely that he will deliver his message.
Similarly, one general may see that someone is leaving the other
general's camp, which leads him to believe that the other general
is trying to send him a message. If, in turn, no message arrives,
it is likely that game $b$ was selected, but that the primary
message was lost. The analysis of such an information structure
that comprises ``primary signals", as in the \citet{Rub89} game
and ``secondary signals", which inform players of the chance that
their primary messages were received, indicates that more general
signal structures can reverse the intended equilibrium selection
effect: multiple equilibria are ensured once private information
is very precise.

In the present example, we find that players face a coordination
game regarding the way that they use their information to reason
about each other's actions. Put differently, the extent to which
player $i$ trembles from the perspective of player $j$ depends on
the way that $i$ weights his information and vice versa.
Accordingly, signals with heterogeneous informational content
induce a coordination game with regard to the way that agents
evaluate their information. In the present specification, there is
a complementarity between player $i$'s primary signal and player
$j$'s secondary signal: if $i$ leans on his primary signal, then
$j$ can forecast $i$'s tremble best through the secondary signal
and vice versa. As we show, this complementarity gives rise to two
classes of equilibria. Moreover, within each equilibrium class,
there exists a countably infinite number of equilibria if the
private signal's precisions are large. That is, while precise
private signals ensure uniqueness in \citet{Rub89} and
\citet{Car93}, they ensure multiplicity in the present class of
games. Finally, the asymmetric equilibria that arise under the
present extended signal structure cannot be dismissed on
efficiency grounds. In the current model asymmetric equilibria
dominate symmetric ones in terms of efficiency.

The current analysis is confined to the simple \citet{Rub89} game
with its discrete information structure. This restriction allows
us to give a constructive proof for the existence of asymmetric
equilibria and it allows us to make a welfare comparison showing
that asymmetric equilibria welfare-dominate their symmetric
counterparts. A generalization to the continuous ``global games"
structure is therefore left for future research.


\newpage
\begin{appendix}


\addcontentsline{toc}{section}{References}
\markboth{References}{References}
\bibliographystyle{apalike}
\bibliography{References}


\section{Proof of Proposition \ref{sym_game}}\label{A0} The proof is
parallel to the one in \citet{Rub89}. First, we establish that
player $i$ plays $A$ when $T_i=0$. Player $i$ considers two
possible scenarios:
\begin{enumerate}
\item With probability $(1-p)$, game $G_a$ is played.

\item With probability $\frac12 p \epsilon$, Player $j$ was
selected, game $G_b$ is played, and the message from player $j$ to
player $i$ was lost.
\end{enumerate}
Hence, we find a lower bound $\tilde{A}$ for i's payoff from
playing $A$ and an upper bound $\tilde{B}$ for i's payoff from
playing $B$:
$$\pi(A) \geq  \frac{(1-p)M + \frac12 p \epsilon 0}{(1-p) + \frac12 p \epsilon} =: \tilde{A}\qquad \qquad \pi(B)  \leq  \frac{-(1-p)L + \frac12 p \epsilon M}{(1-p) + \frac12 p \epsilon} =: \tilde{B}$$
It holds that $\tilde{A}>\tilde{B}$, and thus player $i$ plays
$A$. The induction step from $t-1$ to $t$ is identical to the
original Rubinstein one: assume that both players play $A$ when
they receive a $T_i<t$. Consider that player $i$ gets information
$T_i=t$. For the following argument we denote the probability that
player $i$ was informed first that game $G_b$ is played by
$\kappa_t\in[0,1]$. The posterior probability of player $j$ having
received information $T_j=t-1$ is given by
$$z_t:= \frac{\kappa_t \epsilon + 1- \kappa_t}{\kappa_t (\epsilon+(1-\epsilon)\epsilon) + 1- \kappa_t}>\frac12 \, .$$
In other words, the posterior probability of player $j$ playing
$A$ is larger than $\frac12$, and thus playing $A$ is optimal for
player $i$ as well: Playing $A$ yields 0, while playing $B$ has
expected payoff $z_t(-L)+(1-z_t)M < 0$.


\section{Proof of Proposition \ref{gen_sym_equ}}\label{A1}
We start by proving the first statement. Without loss of
generality, we have to check only if player $i$'s best response to
player $j$'s equilibrium strategy is consistent with player $i$'s
equilibrium strategy. We have to check the following cases of
information that player $i$ might receive:
\begin{enumerate}
\item $Z_i\leq n$: Player $i$ knows that $T_j\leq Z_i \leq n$ and
that $j$ player A, hence plays $A$.
\item $Z_i=n+1$:
\begin{enumerate}
\item $T_i=n$: player $i$ knows that $Z_j=n$ and that player $j$
plays $A$, which leads $j$ to play $A$ as well.
\item $T_i=n+1$: Clearly $Z_j \geq n+1 $ and $T_j\in \{n,n+1\}$. To
determine the payoff of playing $B$, the conditional distribution
of $T_j$ has to be taken into account:
{\footnotesize$$
\frac{(1-\epsilon)^{2n}\epsilon\psi}{(1-\epsilon)^{2n}\epsilon\psi+(1-\epsilon)^{2n+1}\epsilon
(1-\psi)}(-L) + \frac{(1-\epsilon)^{2n+1}\epsilon
(1-\psi)}{(1-\epsilon)^{2n}\epsilon\psi+(1-\epsilon)^{2n+1}\epsilon
(1-\psi)} M \,.$$}
Hence playing $B$ is optimal if $\psi < \psi_3 := \frac{(1-\epsilon)M}{L+(1-\epsilon)M}$.
\item $T_i>n+1$: thus $Z_j\geq n+2$ and $T_j\geq n+2$. Player $j$ plays
$B$ and the optimal response of $i$ is $B$.
\end{enumerate}
\item $Z_i=n+2$: Player $i$ knows that $T_j \geq n+1$.
\begin{enumerate}
\item $Z_i=n+2$ implies that $T_i < n$ is not feasible.

\item $T_i=n$: $Z_j \geq n+1$ with probability $\psi$, hence the payoff
of playing $B$ is $(1-\psi)(-L)+\psi M$, which is negative for
$\psi<\psi_1$, the case when playing $A$ is optimal for player
$i$.

\item $T_i > n$: Therefore $Z_j \geq n+1$ for sure, and thus both
players play $B$.
\end{enumerate}
\item $Z_i>n+2$: Player $i$ knows that $T_j\geq Z_i-1>n+1$. It
also holds that $T_i \geq n+1$ and thus $Z_j \geq n+1$. Therefore player $j$
plays $B$ and player $i$'s best response is to play $B$ as well.
\end{enumerate}
Hence we have established the first part of the proposition for
$\psi \leq \min\{\psi_1,\psi_3\}$.

To prove the second part of the proposition we provide a counter
example: Suppose $Z_1=n+1$ and $T_1=n+2$: player $1$ now plays
$A$. This is not a best response since player $1$ knows
that $T_2 \geq n+1$ and $Z_2\geq n+2$, and thus that $2$ plays $B$ with
certainty. Hence, equilibria where players play $B$ iff $T_i\geq n+1$
and $Z_i\geq n+2$ cannot exist.


\section{Proof of Proposition \ref{p6}}\label{A2}
We will compute the total welfare loss in the asymmetric
equilibria of Proposition \ref{gen_asym_equ} given $n$ (sum of
expected surplus losses of player 1 and 2) compared to
hypothetical perfect coordination between both players. Note that
in state $a$ neither miscoordination nor coordination on the wrong
action can occur. In state $b$
\begin{enumerate}
\item coordination on the wrong action $(A,A)$ happens with
probability
$$p\, \left[1-(1-\epsilon)^{2(n-1)}+(1-\epsilon)^{2(n-1)}\epsilon(1+(1-\psi)(1-\epsilon))+\frac12(1-\epsilon)^{2n}(1-\psi)\right]$$
\item miscoordination $(B,A)$ happens when $T_1=n$, $T_2\geq n$, and
$Z_2=n+1$. The associated probability is $ p(1-\epsilon)^{2n-1} \epsilon \psi$.
\end{enumerate}
Using this, we can compute the welfare loss in equilibrium:
{\footnotesize\begin{multline*}
l_n:= p\, \Bigg\{(1-\epsilon)^{2n-1}\epsilon\psi (2M+L) \\+  \left[1-(1-\epsilon)^{2(n-1)}+(1-\epsilon)^{2(n-1)}\epsilon(1+(1-\psi)(1-\epsilon))+\frac12(1-\epsilon)^{2n}(1-\psi)\right] 2M \Bigg\} \\
= p(1-\epsilon)^{2n-1}\, \Bigg\{\epsilon\psi (2M+L) +
\left[-1+\epsilon(1+(1-\psi)(1-\epsilon))+\frac12(1-\epsilon)^2(1-\psi)\right]
2M \Bigg\} + p2M\, .
\end{multline*}}
It is straightforward to see that the expression in curly brackets
is negative for small $\psi$. Hence, $l_n$ is increasing in $n$
for small $\psi$.

The remaining part of the proof requires computing the welfare
loss in the symmetric equilibria of Proposition \ref{gen_asym_equ}
given $n$. Again there are two types of losses:
\begin{enumerate}
    \item coordination on the wrong action $(A,A)$ happens with probability
    $$p\, \left[1-(1-\epsilon)^{2n}+(1-\epsilon)^{2n}\epsilon(1-\psi)\right]
    $$
    \item miscoordination $(B,A)$ happens once $T_1=n$, $T_2\geq n$, and
    $Z_2=n+1$. Hence it happens with probability $ p(1-\epsilon)^{2(n-1)} \epsilon \psi$.
\end{enumerate}
Using these probabilities, we compute the welfare loss in
equilibrium:
$$\tilde{l}_n = p (1-\epsilon)^{2(n-1)} \left\{ \epsilon \psi (2M+L) + (1-\epsilon)^2\left[-1+\epsilon(1-\psi)\right]2M\right\} +p2M$$
Note that $l_n-\tilde{l}_n \xrightarrow[\psi \to 0]{} -p (1-\epsilon)^{2n} \epsilon M < 0$.
Hence, for a small enough $\psi$ it holds that $l_n <
\tilde{l}_n$, i.e., welfare is higher in the asymmetric equilibria
given $n$.


\section{Referee appendix}\label{ref}

In this appendix we show that our main result, multiplicity of
equilibria in the presence of primary and secondary signals, holds
for the original asymmetric version of the electronic mail game of
\cite{Rub89}. That is, we now assume that it is always player 1
who gets informed first, i.e., gets a message in case nature draws
game $b$. Equivalently, we set the probability $P$, with which
player 1 is informed first to $P=1$ (rather than $1/2$, which is
what we assumed in the main text). Other than that leave the
signals $Z,T$ unchanged. Our only deviation from \cite{Rub89} is
therefore the introduction of the secondary signal $Z$. Naturally, Proposition \ref{gen_orig_equ} holds without modification of the proof.

We now show for this simplified setting that multiple equilibria
exist as in the main text. In particular, we prove that the
asymmetric equilibria described in Proposition \ref{gen_asym_equ}
still exist:

\begin{prop}\label{orig_asym_equ} There exists an asymmetric threshold
equilibria for every $n\in\{1,2,3,\dots\}$ for small enough $\psi$:
player 1 plays $B$ if and only if $T_1\geq n+1$ and player 2 plays $B$
if and only if $Z_2\geq n+1$ and $T_2 \geq n$.
\end{prop}
\begin{proof}
The proof is mostly unchanged compared to the proof of Proposition \ref{gen_asym_equ}. There are two exceptions:
\begin{enumerate}
\setcounter{enumi}{1}
\item Equivalently, now take the behavior of player 1 as given.
\begin{enumerate}
\setcounter{enumii}{2}
\item $Z_2 = n+1$: Here we have to take care of four subcases:
\begin{enumerate}
\setcounter{enumiii}{1}
\item $T_2 = n$: Note that $P(T_1 = n | T_2 = n \wedge Z_2=n+1) = \psi$, and thus, the payoff of playing $B$ is given by $\psi (-L) + (1-\psi) M$. From this we can determine $\bar{\psi}_2:=\frac{M}{L+M}>0$ such that
for all $\psi\leq\bar{\psi}_2$ playing $B$ is optimal for player 2, i.e. where the expected payoff of playing $B$ is non-negative.
\item $T_2 > n$: Hence $T_1\geq n+1$ for sure, player 2 chooses $B$.
\end{enumerate}
\end{enumerate}
\end{enumerate}
Again, we can choose a small enough $\psi$, i.e., $\psi \leq \min
\{\psi_1,\bar{\psi}_2\}$, such that the strategy profile from the
proposition is indeed an equilibrium.
\end{proof}


\end{appendix}


\end{document}